\documentclass{elsarticle}
\usepackage{amsmath,amssymb,amsbsy,amsthm}
\let\emptyset\varnothing
\let\epsilon\varepsilon
\let\phi\varphi

\def\N{\mathbb N}

\def\E{{\operatorname{\bf{E}}}}

\newtheorem{theorem}{Theorem}

\newtheorem{proposition}{Proposition}

\newtheorem{definition}{Definition}

\def\st{\operatorname{St}}
\def\enc{\operatorname{StEnc}}
\def\dec{\operatorname{StDec}}

\journal{Information \& Computation}

\begin{document}

\begin{frontmatter}
 \author[B]{Boris Ryabko}
\author[D]{ Daniil Ryabko}
\address[B]{
Siberian State University of Telecommunications and Information Sciences and \\ Institute of 
Computational Technology  of Siberian Branch of Russian Academy of Science,\\ Kirov str.86 Novosibirsk,  630102, Russia. email: boris@ryabko.net}
\address[D]{INRIA Lille, 40, avenue Halley
Parc Scientifique de la Haute Borne
59650 Villeneuve d'Ascq, France. email: daniil@ryabko.net}
\title {
Constructing Perfect Steganographic Systems  
}

\begin{abstract}
We propose steganographic systems for the case when covertexts
(containers) are generated by
a finite-memory
source  with possibly unknown statistics. The probability
distributions of covertexts with and without hidden information
are the same; this means that the proposed stegosystems are
perfectly secure, i.e.\ an observer cannot determine whether hidden
information is being transmitted. 
The speed of transmission of hidden information can be made arbitrary close to the theoretical limit~--- the Shannon entropy of
the source of covertexts. 
An interesting feature of the suggested  stegosystems is that they do not require any (secret or public) key.

At the same time, we outline some principled computational limitations on steganography. We show that
there are such sources of covertexts, that any stegosystem that has linear (in the length of the covertext) speed 
of transmission of hidden text must have an exponential Kolmogorov complexity. This shows, in particular, that some assumptions on the 
sources of covertext are necessary.
\end{abstract}


\end{frontmatter}

\section{Introduction}
In this work we take an information-theoretic approach to steganography, and construct
perfectly secure steganographic systems for the case of finite-memory sources of covertext.
We also show that some (probabilistic) assumptions on the sources of covertexts are necessary, by demonstrating 
some principled computational limitations on steganography that arise in the absence of such assumptions.

Perhaps the first  information--theoretic approach  to steganography was proposed by Cachin \cite{Cachin:04},
who modeled  the sequence of covertext  by  a memoryless  distribution.
Besides laying out basic definitions of steganographic protocols and their security,
Cachin has constructed a steganographic protocol,  which,
relying on the fact that the probability distribution of covertexts is known, assures that  the distributions of covertexts
with and without hidden
information are statistically close (but, in general, are not equal). For the case of an unknown distribution,
a universal (distribution--free) steganographic system was proposed, in which this property holds only asymptotically
with the size of the hidden message going to infinity.
Distribution-free stegosystems  are of particular practical importance, since in reality
covertexts can be a sequence of   graphical images, instant or email messages,  that is, sources for which
the distribution is not only unknown but perhaps cannot be reasonably approximated.
Cachin has also defined perfectly secure steganographic systems as those for which the probability distribution
of covertexts with and without hidden information are the same.

In \cite{Ryabko:09steg} a perfectly secure universal (that is, distribution-free) steganographic system was proposed, for the case 
of i.i.d. sources of covertexts. 
Here we generalize this construction,   obtaining a  perfectly secure universal  steganographic system for 
a much larger class of sources covertext: that of all $k$-order Markov sources. The only probabilistic characteristic of 
the source that has to be known is the bound $k$ 
on the memory. 

For any stegosystem the next property after its security that is of interest is its capacity.
The capacity of a stegosystem can be defined as the number of hidden bits transmitted per letter of covertext.
We show that our  stegosystem  has the maximal possible capacity:
the number of hidden bits per covertext approaches (with the length of the block growing)
the Shannon entropy of the source of covertexts.
%

Another important feature of our stegosystems is that they do not require a secret key. Thus,
the constructions presented demonstrate that in order to achieve perfect steganographic
security no secret has to be shared between the communicating parties. Clearly, in this case Eve (the observer)
 can retrieve the secret message being transmitted; however, she will not
be able to say {\em whether it is a secret message or a random noise}. This property of
our stegosystems (as indeed their secrecy) relies  on the fact that the secret message
transmitted is indistinguishable from a Bernoulli (i.i.d.) sequence of equiprobable bits (random noise).
This is  a  standard assumption that can be easily fulfilled if Alice uses the Vernam cipher
(a one-time pad) to encode the secret before transmitting. For this, obviously, a cryptographic key
is required. In other words, a secret key can be used to obtain {\em cryptographic security},
but it is not required to obtain {\em steganographic security}, as long as the hidden information
is already indistinguishable from random noise. This also means that the proposed stegosystems
can be directly applied for covert public-key cryptographic communication.

 The main idea behind the stegosystems we propose is the
following. Suppose that for a  covertext $x$ generated by
a source, we can find a set $S$ of covertexts such that each covertext in $S$   has the same
probability of being generated as $x$. Moreover, assume
that each element of $S$   defines $S$  uniquely. Then,
instead of transmitting the covertext $x$ that was actually generated, we can
transmit the covertext in the set $S$ whose  number in $S$ (assuming some pre-defined natural ordering) corresponds
to the secret text we want to pass. This does not change the
probabilistic characteristics of the source, provided the secret 
text consists of i.i.d.\ equiprobable bits. Therefore, an observer
cannot tell whether secret information is being passed. Consider a
simple example. Suppose that Alice wants to pass a single bit, and
assume that the source of covertexts is i.i.d., but its
distribution is  unknown. Alice reads two symbols from the source,
say $ab$. She knows that (since the source is i.i.d.) the
probability of $ba$ is the same. So if Alice's secret bit  to
pass is 0 she transmits $ab$ and if she needs to pass 1 then she
transmits $ba$. However, if the source has generated $aa$ then
Alice cannot pass her secret bit, but she has to transmit $aa$
anyway, to preserve the probabilistic characteristics of the
source. (This example is considered in more details in
Section~\ref{sec:main}.) The same idea was used by  von~Neumann
\cite{Neumann:51} in his method of extracting random equiprobable bits
from a source of i.i.d.\ (but not necessarily equiprobable)
symbols. 
There are two disadvantages of the outlined stegosystem: first, the rate of transmission
of secret text is not optimal, and second, it applies only to i.i.d.\ covertexts.
The following generalization surmounts both obstacles.
First, observe that for a sequence of symbols of length $n$ output
by an i.i.d.\ source (with unknown characteristics), all 
permutations of this sequence have the same probability. To pass secret
information, Alice transmits the permuted sequence whose index number
(in the set of all permutations) encodes her message. A
stegosystem based on this principle achieves (asymptotically with
the block length $n$ growing) maximal possible rate of
transmission of hidden text: the Shannon entropy of the source of
covertexts. Moreover, this idea  works far 
 beyond i.i.d.\ sources of covertexts, by passing from all permutations of a string of covertexts, 
to the set of strings that have the same frequency of occurrence of all tuples of a given length. This way we can  construct
 a stegosystem for $k$-order Markovian sources (for any
given $k$).

Thus, we will show that there is a wide class of sources of covertexts, for which  simple, 
perfectly secure steganographic systems exist.
Naturally, one is interested in the question of
whether such stegosystems exist for any possible (stationary)
source of covertext. This problem is of interest since sources of
covertexts that are of particular practical importance, such as
texts in natural languages or photographs, do not seem to be
well-described by any known simple model (in particular, the finite-memory assumption is often violated). Here we answer this
question in the negative. More precisely, we demonstrate that there
exist such  sets of distributions on  covertexts of  length $n$, for which
 simple stegosystems whose speed of transmission of hidden text  $\Omega(n)$ do not exist. Here  simplicity
is measured by Kolmogorov complexity  of the system, and  stegosystem is considered ``simple''
when its complexity is $\exp( o(n))$, when $n$ goes to infinity.
Kolmogorov complexity  is an intuitive notion which  often helps to establish first
results that help to understand the principled limitations a certain problem
or model imposes; it has been used as such in many works, see, for example, \cite{Vitanyi:08, Uspenskii:90, Vitanyi:00, Vyugin:02}.

This result can be interpreted as that there are such complicated sources
of data that one cannot conceivably put significantly more information into a source
without changing its characteristics, even though the entropy of the source is very high.
This seems to reflect what is well-known in practice; to take one example, it is apparently very hard to put any hidden
message into a given text in a natural language, without making the text ``unnatural.'' Of course, rather
than trying to change a given text,  the communicating
parties can easily agree in advance on two  texts each of which codes one secret bit, so that when the need
for communication arises, Alice can transmit one of the texts, thereby passing one bit. However,
in order to communicate more than one bit, to use the same method they would have to have
a database of covertexts that is exponentially large with respect to the message to pass. Moreover,
even this stegosystem would not be perfectly secure, since the source of covertexts with hidden
information is concentrated on a small subset of all the possible covertexts of given length.
If the stegosystem is used once, then perhaps no reliable detection of the hidden message is possible.
If it is to be used on multiple occasions, that is, if we wish to construct a general purpose stegosystem
for transmitting, say, $\delta n$ bits with an $n$-bit message (for some fixed $\delta>0$), we will need
to construct a database of effectively all possible covertexts. At least, this is the case for some sources
of covertexts, as we demonstrate here, and it seems likely that it is the case for such 
sources as texts in natural languages or even photographic images. Thus, our negative result may be helpful
in clarifying the nature of the difficulties that arise in construction of real steganographic
systems which use human-generated sources of covertexts.

{\bf Contents.} The rest of the paper is organized as follows.
In the next section we present the basic definitions.
In Section~\ref{sec:main} we present a (perfectly secure) stegosystem for finite-memory source of covertexts,
which has the mentioned asymptotic
properties of the rates of hidden text transmission; this stegosystem is a generalization  of
the stegosystem for i.i.d.\ sources of covertexts described in  \cite{Ryabko:07stego, Ryabko:09steg}. 
In Section~\ref{sec:alg} we
briefly describe how this stegosystem can be algorithmically realized in practice.
In Section~\ref{sec:skc} we present a result of the opposite kind: there are  sources of covertexts that
are so complex that any stegosystem that has a linear speed of transmission, must have an exponential Kolmogorov complexity.

\section{Notation and definitions}\label{sec:not}
We use the following model for steganography, mainly following \cite{Cachin:04}.
It is assumed that Alice  has an access to an oracle which generates covertexts
according to some fixed but unknown {\em distribution of covertexts $\mu$}.
 Covertexts belong
to some (possibly infinite) alphabet $A$. Alice wants to use this source
for transmitting hidden messages.
It is assumed that Alice does not know the distribution of covertexts generated by the oracle, but
  this distribution is either memoryless or has a finite memory (or order); moreover, a bound on the memory
of the source of covertexts is known to all the parties (and is used in the stegosystems as a parameter).  

A {\em hidden message} is a sequence of letters from $B=\{0,1\}$
generated independently with equal probabilities of $0$ and $1$. We denote the {\em source of hidden messages} by $\omega$.
This is a commonly used  model for the source of secret messages, since it
is assumed that secret messages are encrypted by Alice using a key shared only with Bob.
If Alice  uses the
Vernam cipher  (a one-time pad) then the encrypted messages are indeed generated according
to the Bernoulli $1/2$ distribution, whereas if Alice uses modern block or stream ciphers
the encrypted sequence ``looks like'' a sequence of random Bernoulli $1/2$
trials. (Here
 ``looks like'' means indistinguishable in polynomial time,
or that the likeness is confirmed experimentally by
 statistical data, see, e.g. \cite{Menzes:96, BRyabko:05}.)
The third party, Eve, is a passive adversary: Eve  is reading all messages passed from Alice to Bob and is trying
to determine whether secret messages are being passed in the covertexts or not.
Clearly,  if covertexts with and without hidden information have the same
probability distribution ($\mu$) then  it is impossible to distinguish them.
Finite groups of (covertext, hidden, secret) letters are sometimes called (covertext, hidden, secret) words.
Elements  of $A$ ($B$) are usually denoted by $x$ ($y$).

The steganographic protocol can be summarized in the following definition.
\begin{definition}[steganographic protocol]
 Alice draws a {\bf sequence of covertexts} $x^*=x_1,x_2,\dots$ generated by a {\bf source of covertexts} $\mu$, where $x_i$, $i\in\N$ belong to some
(finite or infinite) alphabet $A$.

Alice has a sequence $y^*=y_1,y_2,\dots$ of {\bf secret text} generated  by a source $\omega$ of i.i.d.\ 
equiprobable bits $y_i$:  $\omega (y_i=0)=\omega (y_i=1)=1/2$, independently for all $i\in\N$.
The sources $\mu$ and $\omega$  are assumed independent.

A {\bf stegosystem} $St$ is a pair of functions: an encoder and a decoder. The encoder $StEnc$
 is a  function from $A^n\times \{0,1\}^*$ (a block of covertexts and a secret sequence) to $A^n$,
where $n\in\N$ is a parameter (the block length), whose value is known to all parties (including Eve).
The decoder $StDec$ is a function from $A^n$ to $\{0,1\}^*$. Moreover, $StDec(StEnc(x,y))=y$ for all $(x,y)\in A^n\times \{0,1\}^*$ for
which $StEnc(x,y)$ is defined (that is, {\bf decoding is performed without errors}).
It is assumed that $St^n(x,t)$ can be undefined for some values of $(x,t)\in A^n\times \{0,1\}^*$, the interpretation 
being that Alice can chose how many secret bits she can transmit based on the covertext $x$ and the secret text $y^*$, and that she always has more secret 
bits than she can transmit.

From $x^*$ and $y^*$ 
Alice, using a  stegosystem
$St$ obtains a {\bf steganographic sequence} $X=X_1,X_2,\dots$
that is transmitted over a public channel to Bob.
 Bob (and any possible observer Eve) receives $X$ and obtains, using the decoder $StDec(X)$, the  resulting sequence $y^*$.

The {\bf speed of transmission of secret text} $L_n$ is defined as the expected (with respect to the sources of covertexts $x^*$ and secret bits $y_1,y_2,\dots$)
average (per letter of  covertext)
length of the secret message that is transmitted
\begin{equation}\label{eq:ln}
 L_n(St):= {1\over n}\E_{\mu\times\omega} \max \{k\in\{0\}\cup\N: StEnc(x_1\dots x_n,y_1..y_k)\text{ is defined}\} 
\end{equation}

\end{definition}

For the convenience of notation, the definition is presented in terms of an infinite sequence of secret text. 
It means that a stegosystem can use as many or as few bits of the hidden text for transmission in a given block as is needed.
In practice, of course, Alice has only a finite sequence to pass, which may result in that she will run out of secret bits
when transmitting the last block of covertexts. In this case we assume that the end of each message can always be determined
(e.g. there is always an encrypted ``end of message'' sign in the end), so that Alice can fill up the remainder with random noise.
The sequence of covertexts obtained from the source is routinely broken into blocks of size $n$, when $n$ is the parameter of the
stegosystem. For comparison, in the simple stegosystem presented in the beginning of this section we had $n=2$.

Observe that we require by definition of a  steganographic system that the decoding is always correct. Moreover,
we do not consider noisy channels or active adversaries, so that Bob always receives  what Alice has transmitted.

Note also that there is no secret key in the protocol. A secret key may or may not be used before entering into the steganographic
communication in order to obtain the hidden sequence $x^*$; however, this is out of scope of the protocol.

\begin{definition}[perfect security]
 A steganographic system is called (perfectly) secure if the sequence of covertexts $x^*$ and the steganographic sequence $X$ have
the same distribution: $Pr(x_1,\dots,x_n \in C)=Pr(X_1,\dots,X_n
\in C)$ for any
(measurable)
$C\subset A^n$
  and any $n\in\N$, where
the probability is taken with respect to all distributions involved: $\mu$ and $\omega$.
\end{definition}

\section{A universal  stegosystem for $k$-order Markov sources}\label{sec:main}
Before presenting the stegosystem for $k$-order Markov sources, 
we give an  example of a very simple  stegosystem for i.i.d.\ sources. This stegosystem   demonstrates
in a most concise way the main ideas used then  in the general construction.

 Consider a situation in which  the source of covertexts $\mu$ generates i.i.d.
symbols from the alphabet $A=\{a_1,a_2,a_3\}$.
Let, for example, 
\begin{equation}\label{eq:ex}
y^*=01100\dots, \ \ x^*=a_1a_1\ a_2a_3\ a_3a_3\ a_1a_3\ a_2a_2\ a_2a_1\ a_2a_1\ a_3a_2\dots,
\end{equation}
where $y^*$ was generated  by $\omega$ and $x^*$ is a sequence of covertexts generated by $\mu$.
(Spaces between pairs of letters are introduced to facilitate the reading.)

 The symbols of $x^*$ are grouped into pairs (thus, the block length $n$ equals 2 in this example), which are processed sequentially as follows. 
If the current pair is $a_i a_i$, where $i\in\{1,2,3\}$, then this pair is transmitted unchanged to Bob, 
and no secret information is transmitted with  it. If the current pair is $a_i a_j$ with $i\ne j$
then Alice transmits this pair ordered lexicographically (that is, ordered with respect to the ordering $a_1<a_2<a_3$)
if the secret bit to transmit is 0, and she transmits this pair ordered reverse-lexicographically if the secret bit 
is~1. In other words, in the case $i\ne j$ Alice transmits a pair of symbols selected as follows:
$$
\begin{array}{c|cc}
 \ & y=0 & y=1 \\\hline
 i<j & a_ia_j & a_ja_i   \\
 i>j & a_ja_i &  a_ia_j \\
\end{array}
$$
In our example, the sequence~(\ref{eq:ex}) is transmitted as 
$$
X=a_1a_1\ a_2a_3\ a_3a_3\ a_3a_1\ a_2a_2\ a_2a_1\ a_1a_2\ a_2a_3\dots
$$
Decoding is obvious: Bob groups the symbols of $X$ into pairs, ignores all occurrences of $a_ia_i$, 
and changes $a_ia_j$ to $0$ if $i<j$ and to 1 otherwise.
\begin{proposition}\label{pr}
Suppose that a source $\mu$ generates i.i.d. random variables
taking values in $A=\{a_1,a_2,a_3\}$ and let this source be used for encoding
secret messages consisting of a sequence of i.i.d. equiprobable binary
symbols using the method described above.
Then the sequence of symbols output by the stegosystem
obeys the same distribution $\mu$ as the input sequence.
\end{proposition}
The proof is easy to derive; it  is given in \cite{Ryabko:09steg}.  It is also easy to see that the same method
can be used when the alphabet $A$ is any partially ordered set. 
The ordering can also be arbitrary, and can be known to the observer Eve. For example, 
in the case when $A$ is the set of all digital images, one can assume length-lexicographical ordering on $A$. 

\subsection{The general construction}\label{sec:gen}
Next we describe  the general construction of a universal stegosystem which has
the desired asymptotic properties for finite-memory  sources of covertext.
The main idea is as follows. First, the given sequence of covertexts is divided into blocks, say, of length $n>2k$, where $k$ is an upper bound 
 of the memory of
the source $\mu$ of covertexts.
For each block $x=(x_1,\dots,x_n)$, Alice finds all sequences of covertexts of lengths $n$ that have the same probability
as $x$ and also have the same $k$ leading and $k$ trailing symbols (the latter has to be done so that the
probability of the sequence of  blocks as a whole is intact). Then Alice enumerates all these sequences, and transmits the one
whose number codes her hidden text. To find the sequences that have the same probability as the given one,
this probability itself does not have to be known. Indeed,  words  that
have the same number of occurrence of all subwords of length $k+1$ have the same probability, for any
 $k$-order Markov distribution.

We now proceed with a more formal exposition.
\begin{definition}
A  source (of covertexts) $\mu$ is called (stationary) $k$-order Markov, if 
\begin{multline*}
\mu(x_{n+1}=a|x_n=a_n,x_{n-1}=a_{n-1},\dots,x_1=a_1)\\ = \mu(x_{k+1}=a|x_k=a_n,x_{k-1}=a_{n-1},\dots,x_1=a_{n-k+1})
\end{multline*}
for all $n\in\N$ and all $a,a_1,a_2\dots,a_n\in A$.
\end{definition}

As before, Alice needs to  transmit a sequence $y^*=y_1y_2\dots$ of secret binary messages drawn
by an i.i.d.\ source $\omega$ with equal probabilities of $0$ and $1$, while
 a sequence of covertexts $x^*=x_1x_2\dots$ drawn by an (unknown) source $\mu$ from an alphabet $A$ is available. 
It is known that $\mu$ has memory not greater than $k$, where $k>0$ is given.
First we break the sequence $x^*$ into blocks of $n$ symbols each, where  $n>1$ is a parameter.
Each block will be used to transmit several symbols from  $y^*$  (recall that in the
simple stegosystem given in the beginning of this section we had $n=2$ and each block was used to transmit 1 or 0 symbols).
In this general case the following technical problem arises:
the lengths of the blocks of symbols from $x^*$ and from $y^*$ have to be aligned.
 The problem is that the probabilities of blocks from $y^*$ are divisible by powers of $2$, which is not necessarily the
case with blocks from $x^*$.


Let $u$ denote the first $n$ symbols of $x^*$: $u=x_1\dots x_n$
(the first block), and 
let $\nu_u(a_1\dots a_{k+1})$ be the number of occurrences of the subword $a_1\dots a_{k+1}$ in $u$. Define the set $S_u$ as the set 
of all words of length $n$ in which the frequency of each subword of length ${k+1}$ is the same as in $u$, and  whose first
and last ${k}$ symbols are the same as in $u$:
\begin{multline}
S_u=\Big\{v\in A^n:\\ \forall s\in A^{k+1}\ \nu_v(s)=\nu_u(s);\ \forall t\in \{1,\dots,k,n-k+1,\dots,n\}\ v_t=u_t \Big\}.
\end{multline}
Elements of such  sets, without the restriction on the first and last symbols, are known as strings of the same type, see \cite{Csiszar:98}.


Observe   that, if $\mu$ has memory not  greater than $k$, then  $\mu$-probabilities of all members of $S_u$ are
equal. Let there be given some ordering on the set $S_u$ (for
example, lexicographical) which is known to all communicating parties, and let 
$$S_u=\{s_0,s_1,\dots,s_{|S_u|-1}\}$$
with respect to this  ordering.

Denote $m =\lfloor{\rm log}_2|S_u|\rfloor$, where $\lfloor y\rfloor$ stands for the largest integer not greater than $y$.
Consider the binary expansion of $|S_u|$:

\begin{equation}\label{eq:alph}
|S_u|=(\alpha_{m},\alpha_{m-1},\ldots,\alpha_{0}) ,
\end{equation}
   where $\alpha_{m}=1$, $\alpha_j\in\{0,1\}$ , $m > j \geq0$.
In other words,
$$
  |S_u| =   2^m + \alpha_{m-1} 2^{m-1} +  \alpha_{m-2} 2^{m-2}+ ... + \alpha_0.
$$

Denote $\delta(u)$ the index of the word $u$ in the set $S_u$ (with respect to the considered order)
and let  $(\lambda_{m},\lambda_{m-1},\ldots,\lambda_{0})$ be the binary expansion of 
$\delta(u).$  Let  $j(u)$ be the largest number satisfying  $\alpha_j \neq \lambda_j$. 
Alice, having found   $j(u)$, reads  $j(u)$ letters from the source of hidden text  $y^*$;
let $\tau$ be the number whose binary expansion is this sequence of letters.
Alice finds the word $v$ in $S_u$ whose index is 
$\sum_{ j(u) < s \leq m} \alpha_s 2^s + \tau$  and transmits $v$ to Bob (in other words, 
$v$ is the output of the encoder).

The decoding is as follows. Bob, having received $v$, defines $S_v$ (which equals $S_u$),
then finds (in the same way as for encoding) the number $j(v)$ (which is the same for $u$ and $v$: $j(u)=j(v)$)
and $\tau$, and then using $\tau$ he finds $j(v)$ encoded symbols. 

All the subsequent $n$-letter words are encoded and decoded analogously.
Denote $St^k_n(A)$ the described stegosystem.

The $k$-order (conditional) Shannon
entropy $h_m(\mu)$ of a source $\mu$  is  defined as follows:
\begin{equation} h_m(\mu) =-
\sum_{v \in A^m} \mu(v) \sum_{a \in A} \mu(a|\,v) \log
\mu(a|v).
\end{equation}
\begin{theorem}\label{th:mark}
Suppose that an unknown $k$-order Markov source $\mu$   generates  a sequence of covertext
taking values in some alphabet $A$, where $k\ge0$ is known. Let this source be used for encoding
secret messages  consisting of a sequence of i.i.d. equiprobable binary
symbols using the described method $St^k_n(A)$ with $n>1$.
Then
\begin{itemize}
\item[(i)]
the sequence of symbols output by the stegosystem
obeys the same distribution $\mu$ as the input sequence,
\item[(ii)] If the alphabet $A$ is finite
then the average number of hidden symbols per letter $L_n$ goes to the
$k$-order Shannon entropy $h_k(\mu)$ of the source $\mu$ as $n$ goes to infinity.
\end{itemize}
\end{theorem}
\begin{proof}
 To prove (i) observe that  if, as before, $x_1,x_2\dots$ denotes the sequence generated by the
 source of covertexts, and $X_1,X_2,\dots$
the transmitted sequence,  then by construction 
 we have $P(X_1,\dots,X_n)=P(x_1,\dots,x_n)$
where $n$ is the length of the block. For the second  block we have
\begin{multline*}
P(X_{n+1},\dots,X_{2n}|X_1,\dots,X_n) = P(X_{n+1},\dots,X_{2n}|X_{n-k+1},\dots,X_n)\\=P(X_{n+1},\dots,X_{2n}|x_{n-k+1},\dots,x_n)=
P(x_{n+1},\dots,x_{2n}|x_{n-k+1},\dots,x_n)
\\=
P(x_{n+1},\dots,x_{2n}|x_1,\dots,x_n),
\end{multline*}
where the first and the last equalities follow from the $k$-Markov property, the second is by construction (the last $k$ symbols of each block are kept intact),
and the third one  holds because the hidden texts are equiprobable, as are the elements of $S_u$.
The same holds for all the following blocks, thereby establishing the equality of distributions (i).

Let $S_u'$  be the set of all strings of length $n=|u|$ that have  the same $k$-type as $u$, that
is, the same frequencies of subwords of length $k$: $S_u'=\{v\in A^n: \forall s\in A^{k+1}\ \nu_v(s)=\nu_u(s)\}$.
In other words, $S_u'$ is the same as $S_u$ except  the $k$ first and last symbols are not fixed.
Using a result of the theory of types \cite{Csiszar:98}, for  any $u$ for the size of the  set  $S_u'$ we have   $\log |S_u'|= n h_k(P_u) + o(n)$, where $h_k(P_u)$ is
the $k$-th order entropy of the $k$-order Markov distribution $P_u$ defined by the empirical frequencies of the word $u$.
Since the set $S_u$ is not more than a constant times smaller than $S_u'$ we also have $\log |S_u|= n h_k(P_u) + o(n)$.
Moreover, the law of large numbers implies that $h_k (u)\rightarrow h_k(\mu)$ for $\mu$-almost every sequence $u$ as
its size $n$ goes to infinity. Therefore, 
\begin{equation}\label{eq:sp}
 \log |S_u|= n h_k(\mu) + o(n)\ \text{ with $\mu$-probability 1.}
\end{equation}

Furthermore, define 
$\phi := | S_u|/2^m $ and  let $L(S_u)$ be the average number of secret bits transmitted per 
word from  $S_u:$
$$L(S_u) = \frac{1}{ |S_u|} \sum_{i=0}^m \alpha_i i  2^i  .$$ 
We have
\begin{multline*} L(S_u) = \frac{1}{ |S_u|} \sum_{i=0}^m i\alpha_i 2^i   = \frac{1}{ |S_u|}\left( m  
 \sum_{i=0}^m \alpha_i   2^i - \sum_{i=0}^m \alpha_i 2^i (m-i) \right) \\ = 
 m - \left(2^m \sum_{k=0}^m k \alpha_{m-k}2^{-k}\right) > m - 2^{m+1} / |S_u| =  m- 2/\phi
\\ =
\log |S_u| - \log \phi -2/ \phi .
\end{multline*}
Computing the maximum, we find $\log \phi + 2/\phi \leq 2$ for $\phi \in  [1,2]$.
Thus,   $ L(S_u) > \log |S_u| - 2$.  From this and~(\ref{eq:sp}) we obtain the second statement of the  theorem.
\end{proof}

As it was mentioned in the Introduction, the main idea of the stegosystem $St$ is to construct, for each 
given block of covertexts, a set $S_u$ of equiprobable covertexts. The same idea was used in \cite{Ryabko:09steg}
to construct a stegosystem for i.i.d.\ sources $\mu$, with the main difference being in the definition of the sets $S_u$.
In that work we have also obtained non-asymptotic estimates on the speed of transmission of secret text. Such estimates
should be also possible  to obtain  for the case of $k$-order Markov sources, based on the results 
of the theory of types (e.g., \cite{Csiszar:98}), but, for the sake of simplicity, here we only consider asymptotic behaviour of the speed
of transmission.

\subsection{Complexity of encoding and decoding}\label{sec:alg}
Consider the resource complexity of the stegosystem
$St^0_n(A)$, that is, the general construction of the previous section, but for the case of i.i.d.\ sources of covertexts.
The only resource-demanding part of this stegosystem
is  finding the rank of a given block $u$ in the set $S_u$ of all
its permutations, and, vice versa, finding a block given its rank.
(It is clear that all other operations can be performed in linear time.)

Consider this computational problem in some detail. To store all
possible words from the set $S_u$ would require memory of order
 $|A'|^n n \log |A'|$
   bits,
   (where
   $A'\subset A$ is the set of all
 symbols that occur in $u$
  and $n=|u|$; without loss of generality in the sequel we assume $A=A'$),
 which is practically
unacceptable for large $n$. However, there are  algorithms for
solving this problem with polynomial resource complexity. The
first such algorithm,  that uses polynomial memory with the time
of calculation $c  n^2, c> 0,$ per letter,  was proposed in
\cite{Lynch:66} (see also \cite{Davisson:66, SB}). The time of calculation of the
fastest known algorithm is $O( \log^3 n),$ see \cite{BRyabko:98}.

Next we briefly present the ideas behind the algorithm
from~\cite{Lynch:66}. Assume the alphabet $A$ is binary. Let $S$ be the
set of  binary words of length $n$ with $w$ ones. The main
 observation is the following equality, which gives the lexicographical
index number of any  word $v= x_1 \ldots x_n$ $\in S:$
\begin{equation}\label{ld}
 rank(x_1 \ldots x_n) = \sum_{k=1}^{n} x_{k}
\left (\begin {array}{cc}
n-k \\
w -\sum_{i=1}^{k-1} x_{i}
\end{array} \right) ,
\end{equation} where $\left (\begin {array}{cc}
t \\
m
\end{array} \right) = t!/ (m! (t-m)!), $ $0! = 1$ and $\left (\begin {array}{cc}
t \\
m
\end{array} \right) = 0$ if $t<m.$
 The proof of this well-known equality can be found, for example,  in \cite{Krichevsky:93, BRyabko:98}.
As an example, for $n=4, w=2, v= 1010$ we have
$$
 rank(1010) =
\left (\begin {array}{cc}
3 \\
2
\end{array} \right) + \left (\begin {array}{cc}
1 \\
1
\end{array} \right) = 4.
$$

 The computation by (\ref{ld}) can be performed step-by-step based
 on the following obvious identities:
 $$ \left (\begin {array}{cc}
t  \\  p
\end{array}  \right)  =
\left (\begin {array}{cc} t-1 \\
p-1
\end{array}
\right)  \cdot \frac{t}{p}\enskip  ,\quad \left (
\begin  {array}{cc}  t\\
p
\end{array}
\right)=  \left  (\begin  {array}{cc}
t-1 \\
p
\end{array} \right) \cdot \frac{t}{t-p}. $$
A direct estimation of the number of multiplications and divisions
gives  polynomial time of calculations  by (\ref{ld}). The
method of finding a word $v$ based on its rank, as well as a 
generalization to non-binary alphabets, are based on the same
equality  (\ref{ld}); a detailed analysis can be found in
\cite{Krichevsky:93, BRyabko:98}.

For the  general case of $k$-order Markov sources of covertexts,  again,  the only resource-demanding  part
of the stegosystem $St_n^k(A)$ is the enumeration of all the sequences of the same type as a given one.
For the case $k=1$ and the binary alphabet, \cite{Cover:73} proposes an efficient algorithm for this problem.
For $k=1$ and arbitrary alphabet the recent work \cite{Zhou:2011} gives a method  (Lemma~13) of performing such an
enumeration in time $O(n\log^3\log\log n)$. 
(It is worth noting that the work \cite{Zhou:2011} is devoted to the problem of generating uniformly random bits from 
a Markovian source of data, generalizing the von Neumann scheme, which we have also used (Proposition~\ref{pr}) to 
construct a simple  example of a universal stegosystem. Thus the problem of steganography is closely related to the problem
of generating uniformly random bits from a non-uniform source of randomness.)
For the  case  $k>1$,  the problem of finding  polynomial-time algorithms, to the best of the authors' knowledge, 
 remains open.
 We conjecture that efficient algorithms for this case exist as well, based on the results
cited above, and on the consideration that often the $k$-order Markov case can be reduced to 1-order Markov 
by considering windows of size $k$ as states.


\section{Principled computational limitations on steganography}\label{sec:skc}
In this section we abandon all probabilistic assumptions on the source $\mu$ of covertexts, 
and do not consider asymptotic behaviour of stegosystems with respect to the sequence of covertext.
Therefore, it will be convenient to consider  distributions of covertexts $\mu$ as distributions 
on $A^n$, where $n$ is a parameter interpreted as the total number of covertexts output by the source $\mu$.
In other words, we have just one block of covertexts.  (Clearly, if we show that it is impossible (for some sources) to 
preserve the distribution of one block (the first one), then it is also impossible to preserve the distribution of several 
consecutive blocks, no matter what is the probabilistic dependence between them.)
 With this exception, the rest of the protocol is as defined 
in Definition~\ref{pr}.

We next briefly introduce the notion of {\em Kolmogorov complexity}. A
formal definition can be found, for example, in \cite{Vitanyi:08}.
Informally, Kolmogorov complexity of a word $s$ is the length of
the shortest program that outputs $s$. That is, for some universal
Turing machine $U$, we can define the Kolmogorov complexity
$K_U(s)$ of a binary word $s$ as the length of the shortest
program for $U$ that outputs $s$. There are such machines $U$ that
$K_U(s)\le K_{U'}(s)+const$ for every $s$ and every other
universal Turing machine $U'$ (the constant may depend on $U$ and $U'$ but not on $s$). Fix any such $U$ and 
define Kolmogorov complexity $K(s)$  as $K_U(s)$. 
(So, we can say that  Kolmogorov
complexity is defined up to an additive constant.) 
Complexity of
other objects, such as sets of words or programs, can be defined
via simple encodings into words.
  We will use
some simple properties of $K$, such as $K(s)\le
|s|+c$ for any word $s$, whose proofs can be found in e.g.
\cite{Vitanyi:08}. Here it is worth noting that $K(s)$ does
not take into account time or memory  it takes to compute $s$.

\begin{theorem}\label{th:main}
For every $\delta>0$ there is a family indexed by $n\in\N$ of
distributions $P_n$ on $A^n$ with $h(P_n)\ge n-1$, such that every
stegosystem  $\st_n$ whose Kolmogorov complexity  satisfies $\log
K(St_n)= o(n)$ and whose speed of transmission of hidden text
$L_n(\st_n)$ is not less than $\delta$, is not perfectly secure
from some $n$ on.
\end{theorem}


\begin{proof} The informal outline of the proof is as follows. We will construct a sequence of sets
  $X_n$ of words of length $n$ whose  Kolmogorov complexity is
the highest possible, namely $2^{\Omega(n)}$. For each $n\in\N$,
the distribution $P_n$ is uniform on $X_n$. We will then show
that, in order to have the speed of transmission $\delta>0$ a
perfectly secure stegosystem must be able to generate a large
portion of the set $X_n$, for each $n$. This will imply that the
complexity of such a stegosystem has to  be $2^{\Omega(n)}$. The
latter implication will be shown to follow from the fact that, in
order to transmit some information, a stegosystem must replace the
input with some output that could have been generated by the
source; this, for perfectly secure stegosystems, amounts to
knowing at least a large portion of~$X_n$.

Next we present a more formal proof.
Fix $n\in\N$ and let  $X\subset A^{n}$ be any set such that
$|X|=2^{n-1}$ and
\begin{equation}\label{kx}
 K(X)= 2^{n} (1 +o(1)).
\end{equation}
 The existence of such a set can be shown
by a direct calculation of the number of all subsets with
$2^{n-1}$ elements; the maximal complexity is equal (up to a
constant) to the $\log$ of this value.

The distribution $P_n$ is uniform on $X_n$.
Assume that there is a perfectly secure stegosystem $St_n$ for the
family $P_n$, $n\in\N$, and let the speed of transmission of
hidden text be not less than $\delta$. Define the set $Z$ as the
set of those words which are used as codewords $Z:=\{x\in A^n:
\dec(x)\ne\Lambda\}$. Since the expected speed of transmission of
hidden text is lower bounded by $\delta$, we must have $|Z|\ge
\delta 2^{n-1}$ (indeed, since every word codes at most $n-1$
bits, the expected speed of transmission must satisfy
$(n-1)\frac{|Z|}{2^{n-1}}\ge\delta n$). Since $\st$ is perfectly
secure we must have  $Z\subset X$.
Furthermore, define $Z_0$ as the set of words that code those
secret messages that start with 0, and $Z_1$ those that start
with~1:
\begin{equation}\label{z}
 Z_i:=\{x\in A^n: \dec(x)=iu, u\in \{0,1\}^* \}, i\in\{0,1\}.
\end{equation}
 Since $Z=Z_1\cup Z_0$ we must have $|Z_i|\ge |Z|/2\ge {\delta\over2}2^{n-1}$ for some
$i\in\{0,1\}$. Let this $i$ be~1.

 Thus, we have 
$|X\backslash Z_1|\le 2^{n-1}(1-{\delta\over2})$.
 Let us lower-bound the complexity
$K(Z_1|X\backslash Z_1)$ of the set $Z_1$ given $X\backslash Z_1$. Given
the description of $X\backslash Z_1$ and the description of $Z_1$
relative to $X\backslash Z_1$, one can reconstruct $X$. That is why
$K(Z_1|X\backslash Z_1)\ge K(X)- K(X\backslash Z_1)+O(1).$ 
Hence, 
\begin{equation}\label{2}
K(Z_1|X\backslash Z_1) \ge K(X) - \max_{|U|\le2^{n-1}(1 - \delta/2)}
K(U) +O(1).
\end{equation}
 The latter maximal complexity  can be calculated as
follows:
$$
\max_{|U|\le2^{n-1}(1-\delta)} K(U) = \log     {2^n \choose
2^{n-1} (1-\delta/2)} + O(1).
$$
 Applying the Stirling approximation for factorial, we
obtain
$$
\max_{|U|\le 2^{n-1}(1 - \delta/2)} K(U) \le  2^{n}(1-\gamma)
(1+o(1)),
$$
where $\gamma=1-h({2-\delta\over 4},{2+\delta\over 4})$. From this
equality, (\ref{kx}), and (\ref{2})  we get
%
\begin{equation}\label{kz1}
 K(Z_1|X\backslash Z_1)\ge \gamma 2^{n-1} (1+o(1)).
\end{equation}
 We will next show how to obtain
$Z_1$ from $X\backslash Z_1$ and  the stegosystem $\st$, thus
arriving at a contradiction with the assumption that $\log
K(\st)=o(n)$.

For a set $T\subset X$ define
$$
\phi(T):=\{\enc(x,1u): x\in T, u\in \{0,1\}^{n-1}\}.
$$
 Since $\st$ is perfectly
secure, $\phi(T)\subset X$ for every $T\subset X$.
 Let $T_0=X\backslash Z_1$, and $T_k=T_{k-1}\cup\phi(T_{k-1})$.
 Since $X$ is finite and each $T_{k-1}$ is a subset of $T_k$, there must be such $k_0\in\N$ that
 $T_k=T_{k_0}$ for all $k\ge k_0$.
 There are two possibilities: either $T_{k_0}= X$ or $X\backslash T_{k_0}\ne\emptyset$.
  Assume the latter, and define
$Z_1'=X\backslash T_{k_0}$. Then to obtain an element of $Z_1'$ as
an output of the  stegosystem $\st$, the input must be an element
of $Z_1'$ and a secret message that starts with $1$. From this,
and from the fact that the distribution of the output is the same
as the distribution of the input (that is, $\st$ is perfectly
secure), we get
$$
P_n(Z_1')=P_n(Z_1', y=1u)= P_n(Z_1')\omega(1)=P_n(Z_1')/2,
$$
which implies $P_n(Z_1')=0$ and $Z_1'=\emptyset$. Therefore, there
is a $k\in\N$ such that $T_k= X$. This means that a description of
$Z_1$ can be obtained from a description of $X\backslash Z_1=T_0$
and $\st$. Indeed, to obtain $Z_1$ it is sufficient to run $\enc$
on all elements of $T_0$ with all inputs starting with $1$, thus
obtaining $T_1$, and then repeat this procedure until we get
$T_{k+1}=T_k$ for some $k$, wherefrom we know that $T_k=X$ and
$Z_1=T_k\backslash T_0$. Thus,
\begin{equation}\label{contr}
 K(Z_1|X\backslash Z_1)\le
K(\st)+O(1)=2^{o(n)}
\end{equation}
 which contradicts~(\ref{kz1}).
\end{proof}

\subsection*{Acknowledgements}
The authors  are grateful to the anonymous reviewers for their constructive comments on the paper.
Some preliminary results were reported at  ISIT'09 \cite{Ryabko:09skc}.
Boris Ryabko was partially supported by  the  Russian Foundation of Basic Research; grant 09-07-00005.
Daniil Ryabko was partially supported by the French Ministry of Higher Education and Research, Nord-Pas de Calais Regional Council and FEDER through
CPER 2007-2013.


\end{document}